\documentclass[12pt,a4paper]{article}

\usepackage[T1]{fontenc}
\usepackage[utf8]{inputenc}
\usepackage[english]{babel}
\usepackage{lmodern}
\usepackage{microtype}

\usepackage{amsmath,amssymb,amsthm,mathtools}
\usepackage{bbm}
\usepackage{bm}
\usepackage{mathrsfs}

\usepackage{graphicx}
\usepackage{float}
\usepackage{booktabs}
\usepackage{tikz}
\usetikzlibrary{arrows.meta, shapes.geometric, positioning, calc, decorations.pathreplacing, patterns}

\usepackage{geometry}
\usepackage{setspace}
\usepackage{hyperref}
\hypersetup{
  colorlinks=true,
  linkcolor=blue,
  citecolor=blue,
  urlcolor=blue,
  pdftitle={Derivation of the Born Rule in Infinite-Dimensional Operational Theories via Causal Rigidity},
  pdfauthor={Enso O. Torres Alegre}
}

\newtheorem{theorem}{Theorem}
\newtheorem{lemma}{Lemma}
\newtheorem{proposition}{Proposition}
\newtheorem{corollary}{Corollary}
\newtheorem{remark}{Remark}
\newtheorem{definition}{Definition}
\newtheorem{axiom}{Axiom}

\newcommand{\one}{\mathbbm{1}}

\newcommand{\cH}{\mathcal{H}}
\newcommand{\cM}{\mathcal{M}}
\newcommand{\cS}{\mathcal{S}}

\title{\textbf{Causal Rigidity of Born-Type Probability Rules in Infinite-Dimensional Operational Theories}}

\author{
Enso O. Torres Alegre\\
\small ORCID: 0000-0002-6798-8776\\
\small Pontifical Catholic University of Chile, Santiago, Chile\\
\small \texttt{onill@uc.cl}
}

\date{}

\begin{document}

\maketitle

\begin{abstract}
We provide an operational rigidity result for probability rules in infinite-dimensional settings, applicable under normality and steering assumptions. Starting from a topological generalization of generalized probabilistic theories (GPTs), we consider probability rules of the form $P(\phi|\psi) = \Phi(\tau(\psi,\phi))$, where $\tau$ is an intrinsic operational transition probability between operationally pure states. We demonstrate that under three operationally motivated requirements—(i) no superluminal signaling, (ii) availability of normal steering through purification (in a $\sigma$-additive sense), and (iii) $\sigma$-affinity of probability assignments under countable preparation mixtures—the function $\Phi$ must coincide with the identity on $[0,1]$. Any strictly convex or concave deviation yields an operational signaling distinction in steering scenarios, while continuity combined with $\sigma$-affinity excludes non-affine alternatives within this class. This identifies a unique causal fixed point: within the class $P=\Phi\circ\tau$, the Born rule emerges as the only probability law compatible with no-signaling in operational theories admitting normal steering. We connect this operational result to standard infinite-dimensional quantum mechanics via the normal state space of von Neumann algebras and the GNS representation, recovering the conventional Born rule for projection-valued and positive operator-valued measurements. We discuss the scope of the assumptions and implications for proposed post-quantum modifications in continuous-variable and quantum field-theoretic regimes.
\end{abstract}


\vspace{0.5cm}

\section{Introduction}

The Born rule represents a cornerstone of quantum theory, specifying that measurement probabilities for pure states are given by the squared modulus of inner products: $P(i) = |\langle \phi_i | \psi \rangle|^2$ \cite{Born1926}. In finite-dimensional Hilbert spaces, multiple foundational approaches have been developed to derive this quadratic probability rule from more basic principles, including Gleason-type theorems \cite{Gleason1957,Busch2003}, decision-theoretic arguments \cite{Deutsch1999,Wallace2002}, envariance-based derivations \cite{Zurek2003}, and operational reconstructions \cite{Hardy2001,Chiribella2010,Masanes2011,Barrett2007,Janotta2014}.

However, physically significant quantum systems—including continuous-variable systems, quantum fields, and thermodynamic limits—are fundamentally infinite-dimensional. Extending "principle-based" derivations beyond finite dimensions requires more than a naive replacement of finite-dimensional vector spaces with larger counterparts; it necessitates careful handling of topological structure, $\sigma$-additivity, and measure-theoretic aspects of state preparation, transformation, and measurement. The distinction between normal (physically realizable) and singular (pathological) states becomes crucial in this context, mirroring the operator-algebraic framework of von Neumann algebras \cite{vonNeumann1932,KadisonRingroseI,Takesaki1979}. A closely related derivation in the finite-dimensional GPT framework, based on causal consistency and steering, was presented in \cite{TorresAlegre2025}. 

This paper extends the operational causal-consistency approach to deriving the Born rule to infinite-dimensional operational theories. Our central finding is one of \emph{rigidity}: once an operational theory admits normal steering (a $\sigma$-additive generalization of the Hughston-Jozsa-Wootters theorem \cite{Hughston1993}) and requires probabilities to respect countable mixing in a physically regular manner, relativistic no-signaling uniquely fixes the probability functional form to be linear in the operational transition probability.

\subsection{What This Work Provides}

Our contribution advances the field in several specific ways:

\begin{enumerate}
    \item \textbf{Infinite-Dimensional Operational Framework}: We formulate a topological GPT framework suitable for infinite dimensions, emphasizing the physical distinction between normal ($\sigma$-additive) and pathological probability assignments.
    
    \item \textbf{Minimal Assumptions for Infinite Dimensions}: We identify $\sigma$-affinity (normality under countable mixing) as the crucial infinite-dimensional replacement for finite convex affinity, preventing unphysical probability assignments that respect finite mixing but fail under countable limits.
    
    \item \textbf{Causal Rigidity Theorem}: We prove that under no-signaling, normal steering, and $\sigma$-affinity, any probability rule expressible as $P = \Phi \circ \tau$ must satisfy $\Phi(p) = p$ for all $p \in [0,1]$, establishing the Born rule as a causal fixed point.
    
    \item \textbf{Connection to Standard Quantum Formalism}: We demonstrate how the operational result aligns with the mathematical structure of infinite-dimensional quantum mechanics through the normal state space of von Neumann algebras and the Gelfand-Naimark-Segal (GNS) construction.
    
    \item \textbf{Physical Implications}: We show that any post-quantum modification proposing a nonlinear $\Phi$ in continuous-variable or quantum field-theoretic contexts would generically enable superluminal signaling under normal steering scenarios.
\end{enumerate}

\subsection{Roadmap}

Section~\ref{sec:framework} introduces our infinite-dimensional operational framework, defining topological operational systems, normal states and effects, and the relevant convex and measure-theoretic structures. Section~\ref{sec:tau} defines the operational transition probability $\tau$ and establishes its basic properties. Section~\ref{sec:axioms} formulates our three core axioms: no superluminal signaling (NSS), normal steering via purification (NS), and $\sigma$-affinity of probability assignments. Section~\ref{sec:main} presents our main results: Proposition~\ref{prop:signaling} shows that strictly convex or concave deviations from linearity enable signaling, while Theorem~\ref{thm:main} establishes that continuous monotonic $\Phi$ must be the identity under our axioms. Section~\ref{sec:reconstruction} connects the operational result to standard infinite-dimensional quantum theory through von Neumann algebras and the GNS construction. Section~\ref{sec:implications} discusses physical implications, experimental signatures of deviations, and the physical necessity of normality. Section~\ref{sec:limitations} examines limitations, possible extensions, and connections to quantum field theory. Section~\ref{sec:conclusion} concludes.

\section{Infinite-Dimensional Operational Framework}
\label{sec:framework}

We develop a topological generalization of the GPT framework suitable for infinite-dimensional systems \cite{Plavala2023}. Our approach emphasizes the physical distinction between normal (physically realizable) and singular (pathological) states and measurements, corresponding to the mathematical distinction between $\sigma$-additive and merely finitely additive probability assignments.

\subsection{Topological Operational Systems}

\begin{definition}[Topological Operational System]
An \emph{operational system} is described by a triple $(V, V^+, u)$ where:
\begin{itemize}
    \item $V$ is a real locally convex Hausdorff topological vector space.
    \item $V^+ \subset V$ is a closed, proper, generating cone (the \emph{positive cone}).
    \item $u \in V^*$ is a distinguished continuous positive linear functional (the \emph{unit effect}).
\end{itemize}

The \textbf{normalized state space} is defined as:
\[
\Omega := \{\omega \in V^+ : u(\omega) = 1\},
\]
which forms a convex, compact set in the weak-* topology induced by the dual pairing $(V, V^*)$.

\textbf{Effects} are continuous affine maps $e: \Omega \to [0,1]$. By the Riesz-Markov-Kakutani representation theorem and the affine extension property, these correspond to continuous positive linear functionals $e \in V^*$ satisfying $0 \leq e \leq u$ in the dual order.

A \textbf{measurement} is a collection $\{e_i\}_{i \in I}$ of effects with $\sum_{i \in I} e_i = u$, where the summation converges in the appropriate topology (pointwise on $\Omega$ for finite $I$, and in a $\sigma$-additive sense for countable $I$).
\end{definition}

\subsection{Normal vs. Singular States and Effects}

The infinite-dimensional setting introduces a crucial distinction absent in finite dimensions:

\begin{definition}[Normal and Singular Components]
Every state $\omega \in \Omega$ admits a unique decomposition \cite{Edwards1978}:
\[
\omega = \omega_n + \omega_s,
\]
where $\omega_n$ is \textbf{normal} ($\sigma$-additive on effects) and $\omega_s$ is \textbf{singular} (vanishes on all $\sigma$-additive effects but may be nonzero on finitely additive ones).

Physically, only normal states correspond to operationally preparable states through countable approximation procedures (finite energy cutoffs, finite resolution limits, finite time preparation). Singular states represent mathematical idealizations without operational realizability.
\end{definition}

\begin{axiom}[Operational Normality]
We restrict attention to \textbf{normal states} $\Omega_n \subset \Omega$ and \textbf{normal effects} (those continuous with respect to the normal state topology). This restriction captures the physical requirement that probabilities should be $\sigma$-additive under countable operational approximations.
\end{axiom}

\subsection{Operational Purity and Ensembles}

\begin{definition}[Operational Purity]
A state $\psi \in \Omega_n$ is \textbf{operationally pure} if it is an extreme point of $\Omega_n$ and admits no nontrivial convex decomposition into other normal states. We denote the set of operationally pure normal states by $\Omega_{\mathrm{pure}}$.
\end{definition}

\begin{definition}[Normal Ensemble]
A \textbf{normal ensemble} for a state $\omega \in \Omega_n$ is a Borel probability measure $\mu$ on $\Omega_{\mathrm{pure}}$ (with respect to the weak-* Borel $\sigma$-algebra) such that:
\[
\omega = \int_{\Omega_{\mathrm{pure}}} \psi \, d\mu(\psi)
\]
in the weak sense: for every normal effect $e$,
\[
e(\omega) = \int_{\Omega_{\mathrm{pure}}} e(\psi) \, d\mu(\psi).
\]

This definition generalizes finite convex combinations to countable mixtures while respecting the $\sigma$-additivity required for physical realizability.
\end{definition}

The geometric structure of the state space in infinite dimensions, including the decomposition into normal and singular components and the role of pure states, is illustrated in Figure~\ref{fig:state-space}.

\begin{figure}[htbp]
\centering
\resizebox{0.98\textwidth}{!}{%
\begin{tikzpicture}
\draw[thick, fill=blue!10] (0,0) ellipse (3cm and 2cm);
\node at (0,2.3) {$\Omega$ (State Space)};

\filldraw[red] (2.5,0.8) circle (1.5pt) node[above right] {$\psi$};
\filldraw[red] (1.8,-1.2) circle (1.5pt) node[below] {$\phi$};
\filldraw[red] (-2.2,0.5) circle (1.5pt) node[above left] {$\xi$};

\filldraw[blue] (0.5,0.3) circle (1.5pt) node[above right] {$\omega$};

\draw[dashed, blue] (0.5,0.3) -- (2.5,0.8);
\draw[dashed, blue] (0.5,0.3) -- (1.8,-1.2);
\draw[dashed, blue] (0.5,0.3) -- (-2.2,0.5);

\draw[->, thick] (4,0) -- (6,0);
\node at (4.2,0.3) {Decomposition};

\begin{scope}[xshift=7cm]
\draw[thick, fill=green!10] (0,0) ellipse (1.5cm and 1cm);
\node at (0,1.2) {$\Omega_n$ (Normal)};
\filldraw[green] (0,0) circle (1.5pt) node[below] {$\omega_n$};

\draw[thick, fill=red!10] (2.5,0) ellipse (1.5cm and 1cm);
\node at (2.5,1.2) {$\Omega_s$ (Singular)};
\filldraw[red] (2.5,0) circle (1.5pt) node[below] {$\omega_s$};

\draw[->, thick, dashed] (0,0) -- (2.5,0)
node[midway, above] {$\omega=\omega_n+\omega_s$};
\end{scope}
\end{tikzpicture}%
}
\caption{\textbf{State Space Structure in Infinite Dimensions.}
The full state space $\Omega$ decomposes into normal (physically realizable) and singular (pathological) components. Pure states lie on the boundary, while mixed states admit multiple decompositions into pure states via normal ensembles.}
\label{fig:state-space}
\end{figure}
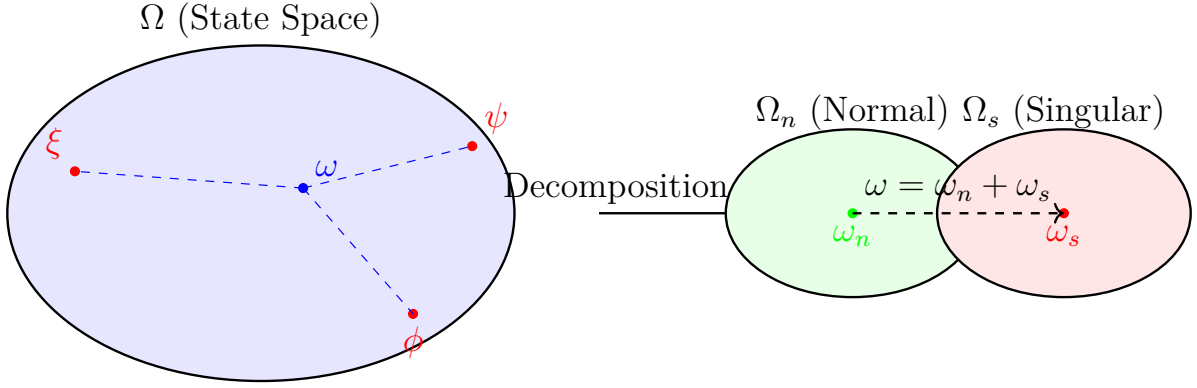

\section{Operational Transition Probability}
\label{sec:tau}

We generalize the concept of transition probability from finite-dimensional GPTs to our infinite-dimensional setting.

\begin{definition}[Operational Transition Probability]
For pure states $\psi, \phi \in \Omega_{\mathrm{pure}}$, define:
\[
\tau(\psi, \phi) := \sup\{ e(\psi) : e \text{ is a normal effect with } e(\phi) = 1 \}.
\]

Intuitively, $\tau(\psi, \phi)$ represents the optimal acceptance probability for state $\psi$ by any test that accepts $\phi$ with certainty. This operational definition does not presuppose any Hilbert space structure.
\end{definition}

\begin{lemma}[Basic Properties of $\tau$]
For all $\psi, \phi \in \Omega_{\mathrm{pure}}$:
\begin{enumerate}
    \item $0 \leq \tau(\psi, \phi) \leq 1$.
    \item $\tau(\phi, \phi) = 1$.
    \item If the theory admits perfectly distinguishable pairs $(\phi, \phi^\perp)$ with a sharp two-outcome test $\{e_\phi, e_{\phi^\perp}\}$, then:
    \[
    \tau(\psi, \phi) + \tau(\psi, \phi^\perp) = 1.
    \]
    \item $\tau$ is symmetric if the effect space is self-dual (a property emerging in quantum theory but not assumed here).
\end{enumerate}
\end{lemma}

\begin{proof}
Properties (1) and (2) follow directly from the definition and the existence of the unit effect $u$ satisfying $u(\phi) = 1$. 

For (3), consider a sharp dichotomic measurement $\{e_\phi, e_{\phi^\perp}\}$ with $e_\phi(\phi) = 1$, $e_{\phi^\perp}(\phi^\perp) = 1$, and $e_\phi(\psi) + e_{\phi^\perp}(\psi) = 1$ for all $\psi$. Since $e_\phi$ satisfies $e_\phi(\phi) = 1$, it is a candidate in the supremum defining $\tau(\psi, \phi)$. For any other effect $e$ with $e(\phi) = 1$, we have $e \geq e_\phi$ on the face generated by $\phi$ and $\phi^\perp$, so $e(\psi) \geq e_\phi(\psi)$. Thus $\tau(\psi, \phi) = e_\phi(\psi)$, and similarly $\tau(\psi, \phi^\perp) = e_{\phi^\perp}(\psi)$, giving the additive relation.
\end{proof}

\begin{remark}[Quantum Specialization]
In standard quantum theory with Hilbert space $\cH$, for pure vector states $|\psi\rangle, |\phi\rangle$, the optimal effect with $\langle \phi| E |\phi \rangle = 1$ is the rank-1 projection $|\phi\rangle\langle\phi|$, yielding:
\[
\tau(\psi, \phi) = |\langle \phi | \psi \rangle|^2.
\]
In the von Neumann algebra setting for normal states, $\tau$ corresponds to the square of the transition probability defined by Uhlmann \cite{Uhlmann1976} through the GNS representation.
\end{remark}

\section{Axioms: Causality, Steering, and $\sigma$-Affinity}
\label{sec:axioms}

We now formulate the three operational principles that will force the Born rule in infinite dimensions.

\subsection{Relativistic Causality}

\begin{axiom}[No Superluminal Signaling - NSS]
For any bipartite system $AB$ with local measurements on subsystems $A$ and $B$, the marginal statistics observed on $A$ are independent of the measurement choice on $B$, and vice versa. Formally, for any normal bipartite state $\omega_{AB}$, any local measurements $\{e_i^A\}$ on $A$ and $\{f_j^B\}, \{g_k^B\}$ on $B$:
\[
\sum_j \omega_{AB}(e_i^A \otimes f_j^B) = \sum_k \omega_{AB}(e_i^A \otimes g_k^B) \quad \forall i.
\]

This axiom implements the core principle of relativistic causality: no faster-than-light communication through measurement choices.
\end{axiom}

\subsection{Normal Steering via Purification}

\begin{axiom}[Normal Steering - NS]
The theory admits bipartite pure normal states $\Psi_{AB} \in \Omega_{\mathrm{pure}}^{AB}$ such that:

\begin{enumerate}
    \item \textbf{Purification}: Every normal state $\omega_B$ on $B$ has at least one purification $\Psi_{AB}$ with $\operatorname{tr}_A(\Psi_{AB}) = \omega_B$.
    
    \item \textbf{Steering Completeness}: For any two normal ensembles $\mu_1, \mu_2$ of the same marginal state $\omega_B$ (i.e., $\int \psi \, d\mu_1(\psi) = \int \psi \, d\mu_2(\psi) = \omega_B$), there exist local measurements on $A$ whose conditional preparations on $B$ realize $\mu_1$ and $\mu_2$ respectively, up to arbitrarily fine operational approximation in the normal state topology.
\end{enumerate}

This axiom generalizes the Hughston-Jozsa-Wootters (HJW) theorem \cite{Hughston1993} to infinite dimensions with $\sigma$-additivity requirements. The "normal" qualification ensures that only physically realizable ensembles (those describable by probability measures on pure states) are steerable.
\end{axiom}

\subsection{$\sigma$-Affinity of Probability Assignments}

\begin{axiom}[$\sigma$-Affinity]
For every normal effect $e$ and every normal ensemble $\mu$ with barycenter $\omega$:
\[
P(e|\omega) = \int_{\Omega_{\mathrm{pure}}} P(e|\psi) \, d\mu(\psi).
\]
Equivalently, probability assignments commute with countable convex combinations in the normal state topology.
\end{axiom}

This axiom captures the physical requirement that probabilities should be stable under countable operational approximation procedures. It excludes pathological finitely additive probability assignments that agree with convex mixing for finite ensembles but fail for countable ones.

\begin{remark}[Physical Necessity of $\sigma$-Affinity]
In laboratory implementations, state preparation always involves finite approximation: truncating infinite sums, discretizing continuous variables, or taking thermodynamic limits. These procedures inherently involve countable limits of operational procedures. $\sigma$-affinity ensures that the probability rule respects these physically necessary approximation schemes, aligning with the measure-theoretic foundations of probability theory \cite{Kolmogorov1933,Holevo2011}.
\end{remark}

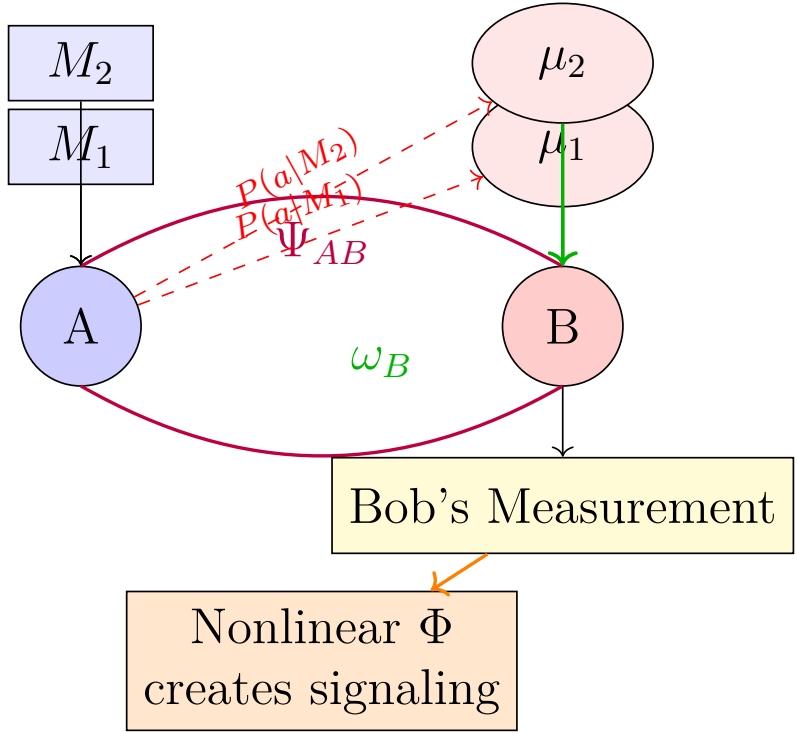
\begin{figure}[htbp]
\centering
\resizebox{0.65\textwidth}{!}{%
\begin{tikzpicture}
\node[draw, circle, fill=blue!20, minimum size=1cm] (alice) at (-2,0) {A};
\node[draw, circle, fill=red!20, minimum size=1cm] (bob) at (2,0) {B};

\draw[thick, purple] (-2,0.5) to[out=30, in=150] (2,0.5);
\draw[thick, purple] (-2,-0.5) to[out=-30, in=-150] (2,-0.5);
\node[purple] at (0,0.7) {$\Psi_{AB}$};

\node[draw, rectangle, fill=blue!10, minimum width=1.2cm, minimum height=0.6cm] (M1) at (-2,1.5) {$M_1$};
\node[draw, rectangle, fill=blue!10, minimum width=1.2cm, minimum height=0.6cm] (M2) at (-2,2.2) {$M_2$};
\draw[->] (M1) -- (alice);
\draw[->] (M2) -- (alice);

\node[draw, ellipse, fill=red!10, minimum width=1.5cm, minimum height=1cm] (E1) at (2,1.5) {$\mu_1$};
\node[draw, ellipse, fill=red!10, minimum width=1.5cm, minimum height=1cm] (E2) at (2,2.2) {$\mu_2$};

\draw[->, dashed, red] (alice) -- node[midway, above, sloped] {\scriptsize$P(a|M_1)$} (E1);
\draw[->, dashed, red] (alice) -- node[midway, above, sloped] {\scriptsize$P(a|M_2)$} (E2);

\draw[->, thick, green!70!black] (E1) to[out=-90, in=90] (bob);
\draw[->, thick, green!70!black] (E2) to[out=-90, in=90] (bob);
\node[green!70!black] at (0.5,-0.3) {$\omega_B$};

\node[draw, rectangle, fill=yellow!20, minimum width=2cm, minimum height=0.8cm] (test) at (2,-1.5) {Bob's Measurement};
\draw[->] (bob) -- (test);

\node[draw, rectangle, fill=orange!20, minimum width=2.8cm, minimum height=1cm, align=center] (nonlin) at (0,-2.8)
{Nonlinear $\Phi$\\creates signaling};
\draw[->, orange, thick] (test) -- (nonlin);
\end{tikzpicture}%
}
\caption{\textbf{Normal Steering and Signaling Test.}
Alice's measurement choice ($M_1$ or $M_2$) steers Bob's system into different ensembles ($\mu_1$ or $\mu_2$) with the same average state $\omega_B$. Under $\sigma$-affinity, if probabilities follow $P=\Phi(\tau)$ with nonlinear $\Phi$, Bob's statistics depend on Alice's choice, enabling superluminal signaling.}
\label{fig:steering-signaling}
\end{figure}

\section{Generalized Probability Rules and Causal Rigidity}
\label{sec:main}

We now investigate probability rules expressible as functions of the operational transition probability.

\begin{definition}[Generalized Probability Rule]
Let $\Phi: [0,1] \to [0,1]$ be a continuous, non-decreasing function with $\Phi(0) = 0$ and $\Phi(1) = 1$. A probability rule is of the form:
\[
P(\phi|\psi) = \Phi(\tau(\psi, \phi)) \quad \forall \psi, \phi \in \Omega_{\mathrm{pure}}.
\]
For mixed states, the rule extends via $\sigma$-affinity: if $\omega = \int \psi \, d\mu(\psi)$, then:
\[
P(\phi|\omega) = \int \Phi(\tau(\psi, \phi)) \, d\mu(\psi).
\]

The Born rule corresponds to the special case $\Phi(p) = p$.
\end{definition}

\subsection{Signaling from Nonlinear Deviations}

The core mechanism linking nonlinearity to signaling involves steering amplification:

The operational mechanism by which nonlinear probability rules lead to superluminal signaling through steering is schematically illustrated in Figure~\ref{fig:steering-signaling}.

\begin{proposition}[Steering Amplification of Nonlinearity]
\label{prop:signaling}
Assume NSS, NS, and $\sigma$-affinity. If $\Phi$ is strictly convex or strictly concave on any nontrivial interval $I \subseteq [0,1]$, then the theory permits superluminal signaling.
\end{proposition}

\begin{proof}
We provide the proof for strict convexity; the concave case follows similarly by reversing inequalities.

Let $\Phi$ be strictly convex on $I$. Choose $p_1 \neq p_2 \in I$ and $\lambda \in (0,1)$ such that $\bar{p} := \lambda p_1 + (1-\lambda) p_2 \in I$. By the richness of pure states implied by NS (operational availability of steering for two-level faces), select pure states $\psi_1, \psi_2, \phi \in \Omega_{\mathrm{pure}}$ such that:
\[
\tau(\psi_1, \phi) = p_1, \quad \tau(\psi_2, \phi) = p_2.
\]

Consider two normal ensembles with the same barycenter $\omega = \lambda \psi_1 + (1-\lambda) \psi_2$:
\begin{enumerate}
    \item $\mu_1 = \lambda \delta_{\psi_1} + (1-\lambda) \delta_{\psi_2}$ (split ensemble).
    \item $\mu_2 = \delta_{\omega}$ (direct preparation).
\end{enumerate}

By NS, Alice can remotely prepare either $\mu_1$ or $\mu_2$ on Bob's system through appropriate measurement choices on her half of a shared entangled state $\Psi_{AB}$, while maintaining the same reduced state $\omega$ on Bob's side.

Let $e_\phi$ be a normal effect achieving (or $\varepsilon$-approximating) the supremum in $\tau(\cdot, \phi)$. By $\sigma$-affinity:
\[
P(e_\phi|\omega; \mu_1) = \lambda \Phi(\tau(\psi_1, \phi)) + (1-\lambda) \Phi(\tau(\psi_2, \phi)) = \lambda \Phi(p_1) + (1-\lambda) \Phi(p_2).
\]

For the direct preparation $\mu_2$:
\[
P(e_\phi|\omega; \mu_2) = P(e_\phi|\omega) = \Phi(\tau(\omega, \phi)).
\]

In the two-level face generated by $\{\psi_1, \psi_2\}$, the operational transition probability behaves affinely in the mixing parameter (this follows from the definition of $\tau$ and the geometry of two-level systems in GPTs):
\[
\tau(\omega, \phi) = \tau(\lambda \psi_1 + (1-\lambda) \psi_2, \phi) = \lambda \tau(\psi_1, \phi) + (1-\lambda) \tau(\psi_2, \phi) = \lambda p_1 + (1-\lambda) p_2 = \bar{p}.
\]

Therefore:
\[
P(e_\phi|\omega; \mu_1) = \lambda \Phi(p_1) + (1-\lambda) \Phi(p_2) > \Phi(\lambda p_1 + (1-\lambda) p_2) = \Phi(\bar{p}) = P(e_\phi|\omega; \mu_2),
\]
where the strict inequality follows from Jensen's inequality for strictly convex $\Phi$.

Thus Bob's measurement statistics for $e_\phi$ depend on whether Alice chose the measurement that steers $\mu_1$ or $\mu_2$, enabling Alice to signal to Bob superluminally by her measurement choice. This violates NSS.
\end{proof}

\begin{remark}[Role of $\sigma$-Affinity]
The proof crucially uses $\sigma$-affinity to equate $P(e_\phi|\omega; \mu_1)$ with the integral $\int \Phi(\tau(\psi, \phi)) d\mu_1(\psi)$. Without $\sigma$-affinity, one could define a probability rule that is finitely additive but not countably additive, potentially evading the Jensen inequality argument for countable mixtures while respecting it for finite ones.
\end{remark}

The emergence of a Jensen gap for convex or concave deviations from linearity, which underlies the signaling argument, is visually represented in Figure~\ref{fig:jensen-gap}.

\subsection{Eliminating Pathological Alternatives}

Proposition~\ref{prop:signaling} rules out strictly convex or concave $\Phi$, but could a continuous non-affine $\Phi$ avoid convexity/concavity on every interval? The combination of continuity, monotonicity, and the functional equations enforced by steering eliminates such pathological cases.

\begin{theorem}[Causal Rigidity of the Probability Rule]
\label{thm:main}
Assume NSS, NS, and $\sigma$-affinity. Let $\Phi: [0,1] \to [0,1]$ be continuous and non-decreasing with $\Phi(0) = 0$, $\Phi(1) = 1$. If $P(\phi|\psi) = \Phi(\tau(\psi, \phi))$ defines operational probabilities compatible with these axioms, then $\Phi(p) = p$ for all $p \in [0,1]$.
\end{theorem}

\begin{proof}
We establish the proof through several steps:

\textbf{Step 1: Local Affinity on Dense Subset.} Proposition~\ref{prop:signaling} shows $\Phi$ cannot be strictly convex or concave on any interval. A continuous function that is neither strictly convex nor strictly concave on any interval must be affine on a dense subset of $[0,1]$ (see \cite{Kuczma2009}, Theorem 15.3.4). Thus there exists a dense set $D \subseteq [0,1]$ and constants $a, b \in \mathbb{R}$ such that $\Phi(p) = ap + b$ for all $p \in D$.

\textbf{Step 2: Extension by Continuity.} By continuity of $\Phi$, the affine relation extends to all $p \in [0,1]$: $\Phi(p) = ap + b$ for all $p \in [0,1]$.

\textbf{Step 3: Boundary Conditions.} Using $\Phi(0) = 0$ gives $b = 0$. Using $\Phi(1) = 1$ gives $a = 1$. Therefore $\Phi(p) = p$ for all $p \in [0,1]$.

\textbf{Step 4: Operational Consistency Check.} With $\Phi(p) = p$, we have $P(\phi|\psi) = \tau(\psi, \phi)$. This rule satisfies $\sigma$-affinity by construction (since integration commutes with the identity function) and does not generate signaling in steering scenarios (as the calculation in Proposition~\ref{prop:signaling} yields equality when $\Phi$ is linear).
\end{proof}

\begin{corollary}[Born-Type Rule at Operational Level]
Under axioms NSS, NS, and $\sigma$-affinity, the unique admissible probability rule of the form $P = \Phi \circ \tau$ is:
\[
P(\phi|\psi) = \tau(\psi, \phi) \quad \forall \psi, \phi \in \Omega_{\mathrm{pure}}.
\]
\end{corollary}

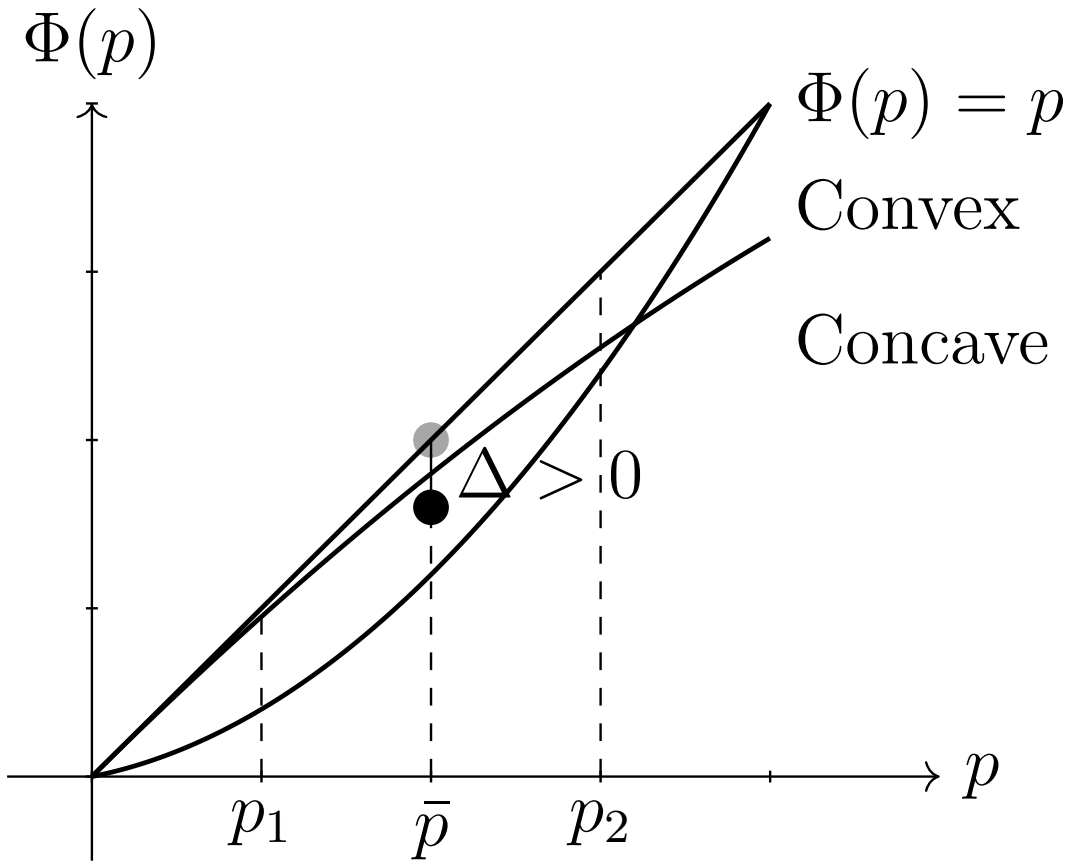
\begin{figure}[htbp]
\centering
\resizebox{0.90\textwidth}{!}{%
\begin{tikzpicture}
\draw[->] (-0.5,0) -- (5,0) node[right] {$p$};
\draw[->] (0,-0.5) -- (0,4) node[above] {$\Phi(p)$};

\draw[thick] (0,0) -- (4,4);
\node[right] at (4,4) {$\Phi(p)=p$};

\draw[thick, domain=0:4, samples=100] plot (\x, {0.2*\x + 0.8*\x*\x/4});
\node[right] at (4,3.4) {Convex};

\draw[thick, domain=0:4, samples=100] plot (\x, {0.8*\x + 0.2*\x*(\x-4)/-4});
\node[right] at (4,2.6) {Concave};

\draw[dashed] (1,0) node[below] {$p_1$} -- (1,1);
\draw[dashed] (3,0) node[below] {$p_2$} -- (3,3);
\draw[dashed] (2,0) node[below] {$\bar{p}$} -- (2,2);

\fill[opacity=0.35] (2,2) circle (3pt);
\fill (2,1.6) circle (3pt);
\draw[dashed] (2,1.6) -- (2,2) node[midway, right] {$\Delta>0$};

\foreach \x in {0,1,2,3,4}
  \draw (\x,1pt) -- (\x,-1pt);

\foreach \y in {0,1,2,3,4}
  \draw (1pt,\y) -- (-1pt,\y);
\end{tikzpicture}%
}
\caption{\textbf{Functional Forms and Signaling.}
The identity $\Phi(p)=p$ corresponds to the Born rule. Convex or concave deviations create a Jensen gap $\Delta$ between $\Phi(\lambda p_1+(1-\lambda)p_2)$ and $\lambda\Phi(p_1)+(1-\lambda)\Phi(p_2)$; via steering this becomes operationally observable, enabling signaling.}
\label{fig:jensen-gap}
\end{figure}

\section{Connection to Infinite-Dimensional Quantum Theory}
\label{sec:reconstruction}

We now demonstrate how the operational result connects to standard infinite-dimensional quantum mechanics through the mathematical framework of von Neumann algebras.

\subsection{Von Neumann Algebraic Representation}

In standard quantum theory, observables of a system form a von Neumann algebra $\cM$ acting on a Hilbert space $\cH$. Normal states correspond to $\sigma$-additive positive linear functionals $\omega: \cM \to \mathbb{C}$ with $\omega(\one) = 1$.

\begin{definition}[Normal State Space of a von Neumann Algebra]
For a von Neumann algebra $\cM$, the normal state space is:
\[
\cS_n(\cM) = \{\omega: \cM \to \mathbb{C} \mid \omega \text{ is positive, linear, normal, and } \omega(\one) = 1\}.
\]
Effects correspond to positive operators $E \in \cM$ with $0 \leq E \leq \one$.
\end{definition}

\begin{theorem}[GNS Representation \cite{Takesaki1979}]
Every normal state $\omega \in \cS_n(\cM)$ induces a cyclic representation $(\pi_\omega, \cH_\omega, \Omega_\omega)$ such that:
\[
\omega(A) = \langle \Omega_\omega | \pi_\omega(A) | \Omega_\omega \rangle \quad \forall A \in \cM.
\]
\end{theorem}

\subsection{Operational Transition Probability Becomes Born Rule}

We now show that in the quantum case, $\tau$ coincides with the standard quantum transition probability.

\begin{proposition}[Quantum Identification of $\tau$]
Let $\cM$ be a von Neumann algebra, and let $\psi, \phi$ be vector states induced by unit vectors $|\psi\rangle, |\phi\rangle$ in a Hilbert space $\cH$. Then:
\[
\tau(\psi, \phi) = |\langle \phi | \psi \rangle|^2.
\]
\end{proposition}

\begin{proof}
The constraint $e(\phi) = 1$ for an effect $E \in \cM$ with $0 \leq E \leq \one$ implies $\langle \phi | E | \phi \rangle = 1$. Since $E \leq \one$, we have:
\[
\langle \phi | (\one - E) | \phi \rangle = 0 \Rightarrow (\one - E)^{1/2} |\phi\rangle = 0 \Rightarrow E|\phi\rangle = |\phi\rangle.
\]
Thus $E$ acts as the identity on the subspace spanned by $|\phi\rangle$. The acceptance probability for $|\psi\rangle$ is $\langle \psi | E | \psi \rangle$. Writing $|\psi\rangle = \alpha |\phi\rangle + \beta |\phi^\perp\rangle$ with $\langle \phi | \phi^\perp \rangle = 0$, we have:
\[
\langle \psi | E | \psi \rangle = |\alpha|^2 \langle \phi | E | \phi \rangle + |\beta|^2 \langle \phi^\perp | E | \phi^\perp \rangle + 2\operatorname{Re}(\alpha^* \beta \langle \phi | E | \phi^\perp \rangle).
\]
Since $E|\phi\rangle = |\phi\rangle$ and $E \geq 0$, the cross term vanishes: $\langle \phi | E | \phi^\perp \rangle = \langle \phi | \phi^\perp \rangle = 0$. Thus:
\[
\langle \psi | E | \psi \rangle = |\alpha|^2 + |\beta|^2 \langle \phi^\perp | E | \phi^\perp \rangle \leq |\alpha|^2 + |\beta|^2 = 1,
\]
with equality achieved when $\langle \phi^\perp | E | \phi^\perp \rangle = 1$, which implies $E|\phi^\perp\rangle = |\phi^\perp\rangle$ as well. The minimal such $E$ (giving the largest possible value for $|\beta|^2 \langle \phi^\perp | E | \phi^\perp \rangle$) is $E = |\phi\rangle\langle\phi|$, yielding:
\[
\tau(\psi, \phi) = \sup_E \langle \psi | E | \psi \rangle = |\langle \phi | \psi \rangle|^2.
\]
\end{proof}

\begin{corollary}[Born Rule in Infinite Dimensions]
In any operational theory satisfying NSS, NS, and $\sigma$-affinity that admits representation by normal states on a von Neumann algebra, the unique causally consistent probability rule is:
\[
P(\phi|\psi) = |\langle \phi | \psi \rangle|^2 \quad \text{(for pure states)}.
\]
For mixed states $\rho$ and POVM measurements $\{E_i\}$:
\[
P(i|\rho) = \operatorname{Tr}(\rho E_i).
\]
\end{corollary}

\section{Physical Implications and Experimental Signatures}
\label{sec:implications}

\subsection{Operational Test for Deviations from the Born Rule}

Our derivation suggests a direct experimental test for any proposed nonlinear modification of the Born rule in infinite-dimensional systems:

\begin{enumerate}
    \item \textbf{Prepare an entangled state} $\Psi_{AB}$ satisfying the normal steering condition.
    
    \item \textbf{Identify a nonlinearity} in the postulated probability rule $\Phi \neq \text{id}$.
    
    \item \textbf{Construct two steering measurements} on $A$ that prepare different ensembles $\mu_1, \mu_2$ of the same reduced state $\omega_B$ on $B$, with the ensembles chosen to maximize the Jensen gap $|\lambda \Phi(p_1) + (1-\lambda)\Phi(p_2) - \Phi(\lambda p_1 + (1-\lambda)p_2)|$.
    
    \item \textbf{Measure statistics} of an appropriate effect $e_\phi$ on $B$ for both steering choices.
    
    \item \textbf{Signal detection}: A statistically significant difference in $P(e_\phi)$ between the two steering scenarios indicates both a deviation from the Born rule and a violation of no-signaling (if the theory maintains the other axioms).
\end{enumerate}

\subsection{Why Normality is Physically Necessary}

The distinction between normal and singular states has direct operational significance:

\begin{itemize}
    \item \textbf{Approximation Procedures}: Any laboratory state preparation involves finite energy, finite time, and finite resolution, corresponding mathematically to taking limits of sequences of normal states. Singular states cannot be approximated in this manner.
    
    \item \textbf{Measurement Statistics}: Real measurements have finite precision and sample finite ensembles, corresponding mathematically to integration against normal effects.
    
    \item \textbf{Thermodynamic Limits}: The infinite-volume limit in statistical mechanics and quantum field theory is carefully controlled to yield normal states on quasi-local algebras \cite{Bratteli1997}.
\end{itemize}

The $\sigma$-affinity axiom captures precisely this physical regularity: probabilities must respect countable approximation schemes that correspond to actual experimental procedures.

\subsection{Implications for Post-Quantum Theories}

Our results place strong constraints on modifications of quantum theory in infinite-dimensional settings:

\begin{enumerate}
    \item Any theory proposing a nonlinear probability rule $P = \Phi(\tau)$ with $\Phi \neq \text{id}$ must either:
    \begin{itemize}
        \item Violate no-signaling (allowing superluminal communication), or
        \item Abandon normal steering (limiting entanglement-assisted remote preparation), or
        \item Reject $\sigma$-affinity (accepting pathological probability assignments).
    \end{itemize}
    
    \item Continuous-variable quantum mechanics and quantum field theory, which are fundamentally infinite-dimensional, are particularly constrained: any nonlinear modification would have observable signaling consequences in steering scenarios.
    
    \item The result provides a \emph{top-down} constraint: even without deriving the full Hilbert space structure, causal consistency alone forces the Born rule once normal steering and countable mixing regularity are assumed.
\end{enumerate}

\section{Limitations and Extensions}
\label{sec:limitations}

\subsection{Technical Assumptions and Their Physical Justification}

\begin{itemize}
    \item \textbf{Normal Steering Axiom}: This is the strongest assumption. While steering is experimentally verified in quantum systems \cite{Wittmann2012}, its extension to all normal ensembles in general operational theories is nontrivial. A weaker but sufficient condition would be steering for a dense set of two-level faces.
    
    \item \textbf{$\sigma$-Affinity}: This requirement aligns with the operational philosophy that only countable approximation procedures are physically realizable. However, some mathematical frameworks consider finitely additive measures; our result shows these lead to signaling if combined with steering.
    
    \item \textbf{Continuity of $\Phi$}: Physical probability assignments should be continuous in operational parameters. Discontinuous $\Phi$ would lead to unphysical instability under small perturbations.
\end{itemize}

\subsection{Extensions to Quantum Field Theory}

The framework naturally extends to algebraic quantum field theory (AQFT) \cite{Haag1996}:

\begin{itemize}
    \item Local observables form a net of von Neumann algebras
    $\{\cM(\mathcal{O})\}$ over spacetime regions.
    \item Physically realizable states are normal on each local algebra.
    \item The steering axiom corresponds to the existence of entangled states across spacelike separated regions.
    \item Our theorem then implies that within any local region, probabilities must follow the Born rule for measurements confined to that region.
\end{itemize}

\subsection{Connections to Other Derivations}

\begin{itemize}
    \item \textbf{Gleason-Type Theorems}: Our approach complements infinite-dimensional Gleason theorems \cite{Bunce1992} by emphasizing causal consistency rather than measure-theoretic structure.
    
    \item \textbf{Decision-Theoretic Approaches}: Like Deutsch \cite{Deutsch1999} and Wallace \cite{Wallace2002}, we derive probabilities from consistency requirements, but our argument works in infinite dimensions and uses steering rather than decision theory.
    
    \item \textbf{Operational Reconstructions}: Compared to finite-dimensional reconstructions \cite{Chiribella2010,Masanes2011}, we specifically address the infinite-dimensional challenges through $\sigma$-affinity and normal steering.
\end{itemize}

\section{Conclusion}
\label{sec:conclusion}

We established an operational rigidity statement for a broad class of probability assignments in infinite-dimensional operational theories. Under three operational requirements—no superluminal signaling (NSS), normal steering via purification (NS), and $\sigma$-affinity under countable preparation mixtures—any probability rule restricted to the form
\[
P(\phi|\psi)=\Phi(\tau(\psi,\phi))
\]
with continuous, non-decreasing $\Phi$ satisfying $\Phi(0)=0$ and $\Phi(1)=1$ is forced to satisfy $\Phi(p)=p$ on $[0,1]$. Equivalently, within this class, operational probabilities coincide with the operational transition probability $\tau$.

When the operational theory admits a representation in terms of normal states on a von Neumann algebra, the operational transition probability reduces to the standard quantum transition probability, and the resulting rule recovers the conventional Born form for both projective measurements and POVMs. In this sense, the Born rule appears as a causal fixed point: within $P=\Phi\circ\tau$, no-signaling together with normal steering and countable-mixing regularity excludes nonlinear distortions of $\tau$.

Several points delimit the scope of the result. First, the steering assumption is the strongest ingredient: the argument requires sufficient normal steering to compare distinct ensemble realizations of the same marginal state. Second, $\sigma$-affinity encodes a physical regularity requirement—stability under countable operational approximations—that rules out finitely additive, operationally unstable probability assignments. Third, our rigidity statement concerns the class $P=\Phi\circ\tau$ and does not claim that every conceivable probability postulate in infinite dimensions must reduce to Born form without additional structure.

These observations suggest a clear constraint on proposed post-quantum modifications in infinite-dimensional regimes: any model that introduces a nonlinear $\Phi\neq \mathrm{id}$ while retaining normal steering and $\sigma$-affinity generically enables operational signaling distinctions in steering scenarios, thereby conflicting with no-signaling. Future work includes weakening the steering requirement (e.g., to a dense family of two-level faces), sharpening the operational conditions under which $\tau$ behaves affinely on relevant faces, and exploring how the present rigidity mechanism interacts with algebraic quantum field theory settings where normality and locality play a central role.

\section*{Acknowledgments}
Parts of this manuscript were prepared with the assistance of a large language model (ChatGPT 5.2) to improve clarity, grammar, and organization of the exposition. All scientific content, including the formulation of assumptions, definitions, results, proofs, and the selection and verification of references, was produced and reviewed by the author, who takes full intellectual responsibility for the work.


\end{document}